\documentclass[11pt]{article}

\usepackage[margin=1in]{geometry}
\usepackage{amsmath,amssymb,amsthm,mathtools}
\usepackage{thmtools}
\usepackage{enumitem}
\usepackage{xcolor}
\usepackage{tikz}
\usetikzlibrary{positioning,calc,arrows.meta,decorations.pathreplacing,backgrounds,fit,shapes.geometric}
\usepackage{float}
\usepackage{algorithm}
\usepackage[noEnd=true]{algpseudocodex}
\usepackage[numbers,sort&compress]{natbib}
\usepackage[colorlinks=true,linkcolor=blue,citecolor=blue,urlcolor=blue]{hyperref}
\usepackage[noabbrev,nameinlink,capitalize]{cleveref}

\declaretheorem[numberwithin=section]{theorem}
\declaretheorem[numberlike=theorem]{lemma}
\declaretheorem[numberlike=theorem]{proposition}
\declaretheorem[numberlike=theorem]{corollary}
\declaretheorem[numberlike=theorem]{observation}

\crefname{algorithm}{Algorithm}{Algorithms}

\newcommand{\OPT}{\mathrm{OPT}}
\newcommand{\MST}{\mathrm{MST}}

\newcommand{\eps}{\varepsilon}
\newcommand{\E}{\mathbb{E}}
\newcommand{\cE}{\mathcal E}
\newcommand{\Prb}{\mathbb{P}}
\newcommand{\tO}{\widetilde{O}}
\newcommand{\RGMIS}{\operatorname{RGMIS}}
\newcommand{\Query}{\mathsf{Query}}
\newcommand{\Member}{\mathsf{Member}}
\newcommand{\Residual}{\mathsf{Residual}}
\newcommand{\rank}{\operatorname{rank}}
\newcommand{\dcup}{\mathbin{\dot\cup}}

\algrenewcommand\textproc{\textsf}
\makeatletter
\providecommand{\theHALG@line}{}
\renewcommand{\theHALG@line}{\thealgorithm.\theALG@line}
\makeatother

\definecolor{DarkGreen}{RGB}{0,100,0}

\title{Improved Sublinear Algorithms \\ for Maximal Independent Set\\
and Metric Steiner Forest}
\author{%
  Sepideh Mahabadi\\
  \small Microsoft Research\\
  \small\texttt{smahabadi@microsoft.com}
  \and
  Jakub Tarnawski\\
  \small Microsoft Research\\
  \small\texttt{jatarnaw@microsoft.com}
}
\date{}

\begin{document}
\maketitle

\begin{abstract}
In this work we consider the Maximal Independent Set (MIS) problem and the metric Steiner Forest problem in the sublinear time setting,
under the adjacency/distance matrix query model.
First, we give an algorithm that estimates the size of an MIS
up to a multiplicative factor of $(1+\eps)$ using $\tO(n^{4/3}/\eps^2)$ queries. This improves the best previous algorithm by Mahabadi, Roghani, Tarnawski, and
Vakilian (SODA 2026), which had a query complexity of $\tO(n^{3/2}/\eps^2)$.

Via a reduction from that work,
this would automatically imply the same improvement
for the problem of estimating the metric Steiner Forest cost
up to an $O(\log n)$ factor.
However, as our second contribution,
we consider the Steiner Forest problem directly
and provide an algorithm with $\tO(n)$ query complexity
that is very simple and does not proceed via MIS.
\end{abstract}

\section{Introduction}
\label{sec:intro}

Sublinear time algorithms form an active research area,
motivated by settings where problem instances are so large that one cannot afford to read the entire input.
There is by now a large body of work on sublinear time algorithms for graph
problems, including minimum spanning tree (MST)~\cite{ChazelleRubinfeldTrevisan05}, maximal independent set (MIS)~\cite{YoshidaYamamotoIto09}, matching~\cite{YoshidaYamamotoIto09,  kapralov2020space, Behnezhad21,BehnezhadRR23a, BehnezhadRRS-SODA23, BhattacharyaKS23, abrrfocs25, MahabadiRT25, BehnezhadRR24, BehnezhadRR23b, BhattacharyaKS-STOC23}, metric MST~\cite{CzumajSohler09}, metric Steiner Tree~\cite{ChenKhannaTan23, MahabadiRTV25}, and metric TSP~\cite{chen2020, ChenMetric-Arxiv22, TSP-icalp24}.
Sublinear algorithms are given {\em oracle
access} to the input.  Two oracle models are standard for graph problems.  In the {\em adjacency list
oracle}, a query $(v,i)$ with $v\in V$ and $i\le n$ returns the $i$-th
neighbor of $v$ (together with the weight of the corresponding edge in the
weighted case), or NULL if $v$ has fewer than $i$ neighbors.  In the {\em
adjacency matrix oracle}, a query $(u,v)$ with $u,v\in V$ returns whether
$(u,v)$ is an edge (again with its weight in the weighted case). In this work we consider the
adjacency matrix model.

\paragraph{Maximal Independent Set.}
An {\em independent set} of a graph $G=(V,E)$ is a set of vertices with no edges between them, and it is {\em maximal} (an MIS) if no
further vertex can be added to it, i.e., if every vertex outside the set
has a neighbor inside it.  While computing a {\em maximum} independent set
is NP-hard to approximate within a factor of
$n^{1-\eps}$~\cite{hastad1996clique,zuckerman2006linear}, a maximal one is
trivial to produce in linear time by a single greedy scan of the vertices.

MIS is a basic problem that has been studied as a problem in its own right, in various computational models such  as distributed and parallel \cite{luby1985simple,alon1986fast,BlellochFS12,ghaffari2016improved,FischerNoever19,GhaffariH21},
Massively Parallel Computing \cite{LattanziMSV11,BrandtFU18,Onak18,ghaffari2018improved,GhaffariU19,GhaffariGJ20},
Local Computation Algorithms \cite{NguyenOnak08,YoshidaYamamotoIto09,RubinfeldTVX11,alon2012space,ReingoldV16,LeviRY17,ghaffari2022local}, or
dynamic~\cite{BehnezhadDHSS19,AssadiOSS19,ChechikZ19}.
It has also been
used as a building block
in algorithms for other fundamental problems,
such as matching and
vertex cover~\cite{YoshidaYamamotoIto09,OnakRonRosenRubinfeld12,Behnezhad21}, or correlation clustering~\cite{AilonCN05,BehnezhadCMT22,DalirrooyfardMS24}.

\paragraph{Sublinear MIS.}
An MIS cannot be found explicitly in sublinear
time\footnote{Take a graph with a single edge, placed on a uniformly
random pair of vertices.  Any MIS must exclude one of its two endpoints, so
locating that edge is unavoidable, and this requires $\Omega(n^2)$ queries
in the adjacency matrix model.}, so the algorithmic task becomes that of
{\em estimating the size} of an MIS, or of designing an oracle that decides
efficiently whether a given vertex belongs to one fixed MIS.

Nguyen and
Onak~\cite{NguyenOnak08} initiated the study of sublinear MIS.
They introduced
the highly influential {\em random-greedy maximal independent set} (RGMIS) notion: fix a permutation
$\pi$ of $V$, scan the vertices in the order $\pi$, and add a vertex to
$\RGMIS(\pi)$ whenever none of its already-scanned neighbors was added.
This process comes with a natural recursive oracle: to decide whether
$v\in\RGMIS(\pi)$, recurse on the neighbors of $v$ that precede it in
$\pi$, and answer yes iff none of them is in.  Yoshida, Yamamoto and Ito~\cite{YoshidaYamamotoIto09} showed that if one queries the earlier
neighbors in increasing order of rank, then for a
random permutation and a random query vertex, the oracle makes
only $O(\Delta)$ recursive calls in expectation, where $\Delta$ is the maximum degree of the graph.
However, this only gives the answer for a single query vertex, and for dense graphs it does not directly yield a truly sublinear time algorithm. Very recently, the first sublinear algorithm for MIS {\em size} estimation was obtained by Mahabadi,
Roghani, Tarnawski, and Vakilian~\cite{MRTV26}: they gave a purely
multiplicative $(1+\eps)$-approximation of $|\RGMIS(\pi)|$ running in
$\tO(n^{3/2}/\eps^2)$ time.

\paragraph{First contribution.}
Our first result speeds up MIS size estimation from $\tO(n^{3/2})$ to
$\tO(n^{4/3})$.  As in~\cite{MRTV26}, the object we estimate is the size of
a random-greedy maximal independent set.

\begin{theorem}[MIS size estimation]\label{thm:mis}
For every $\eps\in(0,1)$ there is a randomized algorithm that, given a
graph $G=(V,E)$ on $n$ vertices with oracle access to its adjacency matrix,
outputs a number $\widehat M$ and a permutation $\pi$ of
$V$ such that, with high probability,
\[
  |\RGMIS(\pi)|\;\le\;\widehat M\;\le\;(1+\eps)\,|\RGMIS(\pi)|,
\]
and runs in $\tO(n^{4/3}/\eps^2)$ time.  (The algorithm does not list
$\RGMIS(\pi)$.)
\end{theorem}

The speed-up rests on {\em pre-solving a random prefix} of the greedy order
(\cref{sec:overview}).  Running greedy on the first $L$ vertices costs only
$O(L^2)$ probes, because -- unlike~\cite{MRTV26} -- we do \emph{not} cross
off the vertices that the prefix solution kills beyond the prefix.
The
price is that the residual graph carrying the rest of the MIS is never
materialized, and has to be handled implicitly.
What makes
that implicit graph tractable
is that
pre-solving the prefix {\em sparsifies} it,
leaving every surviving vertex with residual degree only $\tO(n/L)$.
Moreover, a
{\em relative} form of the same bound covers the vertices that did {\em not}
survive, so that the more oracle traffic such a vertex receives, the cheaper
it becomes to certify that it was discarded.  Balancing the two effects at
$L=n^{2/3}$ gives $\tO(n^{4/3})$.
We give a further high-level explanation in \cref{sec:overview-mis}.

The work of~\cite{MRTV26} uses the sublinear MIS result to get a sublinear time algorithm for the {\em metric Steiner Forest} defined shortly. As a direct corollary of our MIS result, its runtime is improved from $\tO(n^{3/2})$ to $\tO(n^{4/3})$.

\begin{corollary}
\label{cor:sf-via-mis}
Plugging \cref{thm:mis} into the reduction of~\cite{MRTV26} gives an
$O(\log n)$-approximation for metric Steiner Forest in $\tO(n^{4/3})$ time,
with high probability.
\end{corollary}

We note that we further improve this runtime below.

\paragraph{Steiner Forest.}
In the {\em Steiner Forest} problem we are given a weighted undirected graph
$G=(V,E,w)$ with $w:E\rightarrow\mathbb{R}_{\ge0}$ together with a set of
$k$ terminal (or demand) pairs $T=\{(s_1,t_1),\dots,(s_k,t_k)\}$, and we
must find a subgraph of minimum total weight in which every pair
$(s_j,t_j)$ is connected.  The problem generalizes Steiner Tree and is
therefore APX-hard~\cite{ChlebikChlebikova08}.  It has been studied for decades~\cite{AgrawalKleinRavi95, GoemansWilliamson95, GuptaKumar15}, and very recently~\cite{AhmadiGHJM25,GuptaTraub25} improved the long-standing approximation factor of $2$.

\paragraph{Sublinear Steiner Forest.}
Here the
input is a metric space $(V,d)$ on $n$ points, presented as an $n\times n$
distance matrix that the algorithm can query, together with
the $k$ demand pairs; the aim is to estimate the optimal Steiner Forest
cost $\OPT$ using far fewer than the $\binom{n}{2}$ queries needed to read
the whole metric.  Two features of the model are inherent.  First, the
metric assumption is essential: for general distances, even a rough
estimate of the minimum spanning tree weight requires $\Omega(n^2)$
queries~\cite{MRTV26}.  Second, one can only hope to estimate the optimal
\emph{value}: producing an actual near-optimal forest, or even a
spanning tree, already forces $\Omega(n^2)$ queries~\cite{Indyk99}.

Sublinear Steiner {Forest} was studied recently by Mahabadi,
Roghani, Tarnawski, and Vakilian~\cite{MRTV26}, who gave an
$O(\log k)$-approximation running in $\tO(n^{3/2})$ time.

\paragraph{Second contribution.}
The existing sublinear algorithms for Steiner Forest
(\cite{MRTV26} and \cref{cor:sf-via-mis})
proceed via a reduction to MIS.
It is natural to ask whether this is necessary.
Our second result shows that
the optimal Steiner Forest cost
can be estimated with the same $O(\log k)$ approximation
in
{\em near-linear} time
by a short reduction to metric minimum spanning
tree estimation, for which it invokes the algorithm
of Czumaj and Sohler~\cite{CzumajSohler09}.

\begin{theorem}[Steiner Forest]\label{thm:sf}
There is a randomized algorithm that, given oracle access to the distance
matrix of an $n$-point metric space and a list of $k$ demand pairs, outputs
with high probability an $O(\log k)$-approximation of the Steiner Forest
cost, in $\tO(k)$ time.
\end{theorem}

The running time is optimal up to logarithmic factors,
as it is easy to see that no meaningful approximation is possible in $o(k)$ time.

We give a high-level overview of the idea in \cref{sec:overview-sf}.

\paragraph{Concurrent work.}
In concurrent and independent work, Kiss and
Kooroshnezhad~\cite{KissK26} obtained a $(1+\eps)$-approximation
to the size of a random-greedy MIS in the adjacency matrix model
in $\tO(n^{4/3}/\eps^2)$ time.
Via the reduction of~\cite{MRTV26}, this also gives an
$O(\log k)$-approximation for metric Steiner Forest in $\tO(n^{4/3})$
time, as in \cref{cor:sf-via-mis}.
Their algorithm follows a similar high-level strategy to ours:
it partially computes a random-greedy MIS and then simulates
the oracle of~\cite{YoshidaYamamotoIto09} on the implicit residual graph.
Their preprocessing stops after $\Theta(n^{4/3})$ queries,
whereas ours processes a fixed-length prefix.
To test whether a vertex is dominated, they inspect the partial
MIS in greedy order; their analysis applies the residual
maximum-degree bound of~\cite{BehnezhadDHSS19} at exponentially increasing sizes
of the partial MIS. We instead inspect the prefix MIS in a
fresh random order and use the relative-degree bound of
\cref{lem:reldeg}.
They additionally give an application to bi-criteria metric $k$-center
estimation.  On the other hand, they do not consider Steiner Forest beyond
the MIS-based reduction; our \cref{thm:sf} bypasses MIS altogether and
achieves $\tO(k)$ time.

\section{Technical Overview}
\label{sec:overview}

\subsection{A faster MIS size estimator}
\label{sec:overview-mis}
The $\tO(n^{3/2})$-time estimator of~\cite{MRTV26}
works by sampling a permutation $\pi$
and then
repeatedly sampling a \emph{root} vertex $v$
and deciding whether $v \in \RGMIS(\pi)$
using the recursive oracle of Yoshida, Yamamoto, and Ito~\cite{YoshidaYamamotoIto09},
which
recurses on the earlier neighbors of $v$ in increasing order, and answers yes
iff none of them is in.
In the adjacency matrix model,
apart from traversing edges,
one also needs to pay for querying non-neighbors.
The key technical ingredient of~\cite{MRTV26}
is a proof that the natural adjacency-matrix implementation has runtime $O(n)$
(see \cref{lem:published}(ii)).

The other main point of their strategy is to manually find the first $n^{1/2}$
vertices of $\RGMIS(\pi)$; whenever they find the next vertex, they explicitly
cross off all its surviving neighbors (in $O(n)$ time).
If $|\RGMIS(\pi)| \le n^{1/2}$, this constructs the entire MIS;
on the other hand, having $|\RGMIS(\pi)| \ge n^{1/2}$
means that $(1+\eps)$-estimating the MIS size now only requires sampling $\tO(n^{1/2})$ root vertices $v$.
The oracle can run on the explicitly materialized residual graph.
This gives their $\tO(n^{3/2})$ runtime.

In our approach,
we also pre-solve a prefix of $\pi$ of length $L$,
obtaining an MIS $S$ in the prefix.
However, crucially,
we \emph{do not} cross off the MIS neighbors beyond the prefix.
Thus our pre-solving phase runs in $O(L^2)$ time,
and we can afford to set $L=n^{2/3}$ and obtain $n^{4/3}$ runtime.
The remaining task then becomes to estimate MIS size in a \emph{residual} suffix subgraph $R$.
See \cref{fig:prefix} for an illustration.
The main difficulty of our approach
is that this residual graph is now \emph{implicit}.

The first key observation we make is that pre-solving the prefix
has a \emph{sparsification} effect on $R$.
Namely, with high probability,
the degree inside $R$ is at most $D := \tO(n/L) = \tO(n^{1/3})$.
To see this, consider any vertex $y \in V$,
and reveal the prefix vertices one by one.
As long as $y$ has high degree in the hitherto-residual graph
(which is the entire $V$ at the beginning and gradually shrinks),
at each of the $L$ steps
there is a high probability that a neighbor of $y$ will be selected into the MIS
and remove $y$ from the current residual graph.
Eventually either $y$ gets removed, or its residual degree decreases below $D = \tO(n/L)$.

As another observation,
this also implies that $|R|$ cannot be too large compared to the MIS size:
since every MIS-vertex in $R$ can disqualify only $D$ other $R$-vertices,
we have
\begin{equation}
    \label{eq:R-is-small}
    |R| \le (D+1) |\RGMIS(\pi)| .
\end{equation}
And the larger the MIS, the fewer (root vertex) samples we need.

Let us now look at the cost of running the oracle on the implicit residual graph.
If it ran on $G[R]$ that was explicitly materialized as an adjacency matrix,
the \cite{MRTV26} analysis (\cref{lem:published}(ii))
would give $O(|R|)$ runtime.
To get a multiplicative MIS size estimation, one needs $n/|\RGMIS(\pi)|$ (root vertex) samples,
of which $|R|/|\RGMIS(\pi)|$ will be in $R$;\footnote{The non-$R$ root vertex samples take negligible time.}
by \eqref{eq:R-is-small}, this is $O(D)$,
for a total of $O(|R|D)$.

However, since $R$ is implicit, when the oracle at a vertex $x \in R$ loops over
previous vertices $y$ in the permutation,
it incurs the additional costs of
\textit{(a)}~scanning non-edges $(x,y) \not \in E$ (for $y$ both in and not in $R$)
and
\textit{(b)}~checking, for each $y$ with $(x,y) \in E$, whether $y \in R$.
For~\textit{(a)},
the above analysis already counts the costs of scanning non-edges where $y \in R$;
since the positions of residual vertices in the permutation suffix are random,
we just pay an additional factor of $n/|R|$, for a total of $O(|R|D \cdot n/|R|) = O(Dn) = \tO(n^{4/3})$.

For~\textit{(b)}, our accounting is per-$y$.
It is known that the traffic a vertex receives from the oracle
is proportional to its degree:
$y$ is visited $O(|N(y) \cap R|/|R|)$ times
(see \cref{lem:calls}).
On each such visit, we must check whether $y \in R$.
We do this by checking, for every $s \in S$
(recall that $S$ is the prefix MIS),
whether $(s,y) \in E$.
The more $S$-neighbors $y$ has, the faster we can return $y \not \in R$:
this residual test takes expected time
$O(|S|/(1+|N(y) \cap S|))$.\footnote{To avoid probabilistic dependence issues, we iterate over $S$ in a fresh random order; see \cref{lem:cert-cost}.}
As before, we need to take $|R|/|\RGMIS(\pi)|$ root vertex samples from $R$.
Multiplying the three factors yields
\[
    \frac{|N(y) \cap R|}{|R|} \cdot \frac{|S|}{1+|N(y) \cap S|} \cdot \frac{|R|}{|\RGMIS(\pi)|}
    \le
    \frac{|N(y) \cap R|}{1+|N(y) \cap S|}
\]
where we used that $S \subseteq \RGMIS(\pi)$.
Now, for $y \in R$, this is at most the maximum degree of $G[R]$, so at most $D$ by the sparsification property above.
(Summing over $y$ we get $O(Dn) = \tO(n^{4/3})$.)
The case $y \not \in R$ remains.

The final key insight is that a \emph{relative-degree} version of the sparsification property
still holds even for $y \not \in R$:
we have $\frac{|N(y) \cap R|}{1+|N(y) \cap S|} \le D$
with high probability (\cref{lem:reldeg}).
(When revealing the prefix vertices one by one,
as long as this ratio is high for $y$,
the probability that a residual neighbor will be selected next is high;
this allows us to control the trajectory of the ratio.)
Thus, the more traffic a vertex $y \not \in R$ gets,
the less each visit costs;
and we get the same bound for $y \not \in R$
as for $y \in R$.

\subsection{Steiner Forest via capped spanning trees}
\label{sec:overview-sf}
We start with the framework of~\cite{MRTV26}, where the demand pairs are assigned to $O(\log k)$ levels based on their distance. That is, for a threshold $\tau_i$, a pair $(s_j,t_j)$ is assigned to level $i$ if $\tau_i \leq d(s_j,t_j)< 2\tau_i = \tau_{i+1}$. Then, roughly speaking, a budget of $M_i\tau_i$ is used to connect all the demand pairs assigned to level $i$, where $M_i$ is defined as the size of the MIS of some auxiliary graph $H_i$ defined on the clusters formed by satisfying all demand pairs of levels less than $i$. Then a sublinear MIS algorithm is used to estimate $M_i$, and produce an estimate for the total Steiner Forest cost by aggregating the estimates of all levels $i$.

\medskip

In our work, in contrast to \cite{MRTV26}, we do not estimate $M_i$ directly. While we still use it in our analysis, we instead estimate the \emph{minimum spanning tree} (MST) cost of a similar auxiliary graph $H_i'$ using the estimator of Czumaj and Sohler~\cite{CzumajSohler09}.
However, the distances we use are capped by the threshold $\tau_i$. More precisely, we replace each distance $d(u,v)$ by
$\rho_i(u,v)=\min\{d(u,v),\tau_i\}$
(which is again a metric), before passing it to the algorithm of~\cite{CzumajSohler09}.

Let $W_i$ be the weight of the MST (on the set of terminals) in this capped metric. On the one hand, a capped spanning tree of the graph can be extracted
from the optimal forest at a factor-$3$ loss, so $W_i\le 3\OPT$ at every
scale $i$.  On the other hand, the nodes in the MIS have pairwise distances at least $2\tau_i$
and thus even after capping, they still have pairwise distances $\tau_i$, which
forces $W_i\ge\Omega(M_i\tau_i)$. Together with the bound $\OPT\le O\!\big(\sum_i M_i\tau_i\big)$ of~\cite{MRTV26} this makes the
capped MSTs collectively large enough to recover the optimum.  Thus the
weights $W_i$, ranging over $O(\log k)$ scales, sandwich $\OPT$ to
within an $O(\log k)$ factor, and no MIS needs to be computed (nor its size estimated).

\section{Faster MIS Size Estimator}
\label{sec:mis}

In this section we prove \cref{thm:mis}.
Let $G=(V,E)$ be a simple undirected graph
with $|V|=n$, accessed by adjacency-matrix probes ``is $(u,v)\in E$?''.
For a
permutation $\sigma$ of the vertices of a graph $H$, the set
$\RGMIS_H(\sigma)$ is obtained by scanning the vertices in the order
$\sigma$ and selecting a vertex when none of its previously selected
neighbors was selected.
(We skip the subscript $H$ when clear from the context.)

\subsection{The random-greedy oracle and its known bounds}
\label{sec:oracle}

Membership in $\RGMIS_H(\sigma)$ is decided by the standard
recursive oracle~\cite{YoshidaYamamotoIto09}.  To evaluate $\Query(v)$, scan the vertices
preceding $v$ in $\sigma$ in increasing order; on meeting an earlier
neighbor $u$, recursively evaluate $\Query(u)$, returning
\textbf{false} as soon as one such call returns \textbf{true}, and
returning \textbf{true} if none does.  Then
$\Query(v)= $ \textbf{true} iff $v\in\RGMIS_H(\sigma)$.

Our analysis uses two previously established bounds for this
oracle.

\begin{theorem}[oracle bounds]
\label{lem:published}
Let $H$ be a graph on $r\ge 1$ vertices, let $\sigma$ be a uniform
permutation of $V(H)$, and let $v$ be an independent uniform root vertex.  For
the recursive oracle above:
\begin{enumerate}[label=\textup{(\roman*)},leftmargin=2.4em]
  \item for a fixed edge $e$ of $H$, the expected number of times
  the evaluation from $v$ traverses $e$, with multiplicity, is at most
  $1/r$;
  \item in the adjacency-matrix implementation that scans the preceding
  positions of $\sigma$ to find earlier neighbors, the expected total
  number of scanned positions, over all recursive calls, is $O(r)$.
\end{enumerate}
\end{theorem}
\begin{proof}
Part~(i) follows from Lemma~2.3 of~\cite{YoshidaYamamotoIto09}.
A different proof was also given in~\cite{DalirrooyfardMM26}.
Part~(ii) is Lemma 3.8 of~\cite{MRTV26}.
\end{proof}

\subsection{The random prefix and the residual graph}
\label{sec:prefix}

Fix an integer $L\in\{1,\dots,n\}$, to be chosen later.  Draw a uniform
permutation $\pi$ of $V$ and let $P=\{\pi(1),\dots,\pi(L)\}$ be its
\emph{prefix}.  Query every unordered pair of vertices in $P$
(this uses $O(L^2)$ probes and touches no pair with an endpoint outside
$P$), run greedy MIS on $P$ in the order given by $\pi$, and let $S$
be the selected set.  Put
\[
  R=V\setminus N[S],
\]
where $N[S]$ denotes the closed neighborhood of $S$ (i.e., $S \subseteq N[S]$).  Every vertex of $P$ is
either selected into $S$ or rejected by an earlier selected neighbor, so
$P\cap R=\varnothing$.
See \cref{fig:prefix} for an illustration.
We start with the following easy observation.

\begin{observation}[prefix decomposition]
\label{lem:decomp}
Let $\tau=\pi|_R$ be the order that $\pi$ induces on $R$.  Then
$\RGMIS_G(\pi)=S\dcup\RGMIS_{G[R]}(\tau)$.  In particular the right-hand
side is a maximal independent set of $G$ for every $\pi$.
\end{observation}

For $y\in V$ put $a_y=|N(y)\cap S|$ and $b_y=|N(y)\cap R|$. The following
is the key sparsification property
from which the random prefix lets us benefit.

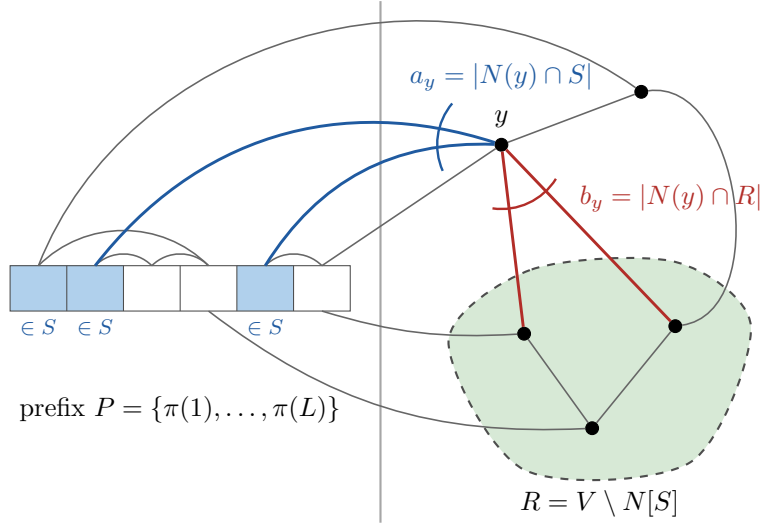
\begin{figure}[t]
\centering
\definecolor{misBlue}{RGB}{31,90,160}
\definecolor{misRed}{RGB}{178,45,40}
\definecolor{misLightBlue}{RGB}{176,208,236}
\definecolor{misLightGreen}{RGB}{198,224,196}
\begin{tikzpicture}[
  vtx/.style={circle,draw,fill=black,inner sep=1.7pt},
  aedge/.style={misBlue,line width=1.1pt},
  bedge/.style={misRed,line width=1.1pt},
  gedge/.style={black!60,line width=0.6pt}]

  \foreach \i in {1,2,5}
    \fill[misLightBlue] ({-6.2+0.75*(\i-1)},-0.3) rectangle ({-6.2+0.75*\i},0.3);
  \draw[black!70] (-6.2,-0.3) rectangle (-1.7,0.3);
  \foreach \i in {1,...,5}
    \draw[black!70] ({-6.2+0.75*\i},-0.3) -- ({-6.2+0.75*\i},0.3);
  \foreach \i in {1,...,6}
    \coordinate (t\i) at ({-6.2+0.75*(\i-0.5)},0.3);
  \foreach \i in {4,6}
    \coordinate (u\i) at ({-6.2+0.75*(\i-0.5)},-0.3);
  \foreach \i in {1,2,5}
    \node[misBlue,font=\scriptsize] at ({-6.2+0.75*(\i-0.5)},-0.52) {$\in S$};
  \node[font=\small] at (-3.95,-1.6) {prefix $P=\{\pi(1),\dots,\pi(L)\}$};
  \draw[gedge] (t1) to[bend left=45] (t4);
  \draw[gedge] (t2) to[bend left=45] (t3);
  \draw[gedge] (t3) to[bend left=45] (t4);
  \draw[gedge] (t5) to[bend left=45] (t6);

  \draw[black!35,line width=0.8pt] (-1.3,-3.1) -- (-1.3,3.8);

  \begin{scope}[on background layer]
    \fill[misLightGreen!70] plot[smooth cycle] coordinates
      {(-0.4,-0.7) (0.7,0.25) (2.6,0.35) (3.6,-0.5) (2.8,-2.3) (0.6,-2.4)};
  \end{scope}
  \draw[dashed,black!70,thick] plot[smooth cycle] coordinates
    {(-0.4,-0.7) (0.7,0.25) (2.6,0.35) (3.6,-0.5) (2.8,-2.3) (0.6,-2.4)};

  \node[vtx,label={[font=\small]above:$y$}] (y) at (0.3,1.9) {};
  \node[vtx] (z)  at (2.15,2.6) {};
  \node[vtx] (r1) at (0.6,-0.6) {};
  \node[vtx] (r2) at (2.6,-0.5) {};
  \node[vtx] (r3) at (1.5,-1.85){};

  \draw[gedge] (y) -- (z);
  \draw[gedge] (t1) to[bend left=52] (z);
  \draw[gedge] (t6) to (y);
  \draw[gedge] (u6) to[bend right=10] (r1);
  \draw[gedge] (u4) to[bend right=20] (r3);
  \draw[gedge] (r1) -- (r3);
  \draw[gedge] (r2) -- (r3);
  \draw[gedge] (r2) to[bend right=90] (z);

  \draw[aedge] (t2) to[bend left=35] (y);
  \draw[aedge] (y) to[bend right=30] (t5);
  \draw[bedge] (y) -- (r1);
  \draw[bedge] (y) -- (r2);

  \draw[misBlue,line width=0.9pt] (0.3,1.9) ++(140:0.85) arc (140:204:0.85);
  \draw[misRed,line width=0.9pt]  (0.3,1.9) ++(262:0.85) arc (262:328:0.85);

  \node[misBlue,font=\small] at (0.3,2.8) {$a_y=|N(y)\cap S|$};
  \node[misRed,font=\small,anchor=west] at (1.2,1.2) {$b_y=|N(y)\cap R|$};
  \node[font=\small] at (1.6,-2.85) {$R=V\setminus N[S]$};
\end{tikzpicture}
\caption{We greedily pre-solve the length-$L$ prefix $P=\{\pi(1),\dots,\pi(L)\}$ (the array on the left).
Blue shaded cells denote the
selected set $S$; every other cell is rejected by an edge from an earlier
selected cell.
To the right of the long vertical line is the suffix;
the \emph{residual graph} $R$ (green shaded) are those suffix vertices that have no edge to $S$.
\\
In the analysis
we often condition on the prefix $P$
while keeping the suffix order  random.
Note that $S$ and $R$ are determined by the prefix.
The algorithm does not explicitly materialize $R$.
\\
For the suffix vertex $y$ shown here,
we illustrate its relative degrees $a_y$ and $b_y$.
\cref{lem:reldeg} says that $b_y \le D (a_y+1)$ with high probability. This implies that $G[R]$ has degree $\le D$.}
\label{fig:prefix}
\end{figure}

\begin{lemma}[relative-degree bound]
\label{lem:reldeg}
With
\[
  D=\min\!\Big\{n,\ \frac{30n\log n}{L}\Big\},
\]
the prefix satisfies $b_y\le D(a_y+1)$ for every $y\in V$ with probability
at least $1-n^{-13}$.  On this event, the maximum degree of $G[R]$ is at most $D$.
\end{lemma}
\begin{proof}
The case $n=1$ is immediate, so assume $n\ge2$.
Fix $y$.
We expose the first $L$ positions one at a time.  After $t$ positions have
been processed, let $S_t$ be the vertices selected so far and
$R_t=V\setminus N[S_t]$.  Set $A_t=1+|N(y)\cap S_t|$ and
$B_t=|N(y)\cap R_t|$.
We define $\cE_t$ to be the event that $B_t > DA_t$.
Our goal is to show that $\Prb(\cE_L)$ is very small.
To that end,
we analyze the random variable $Z_t=\mathbf 1_{\cE_t}/A_t$
and how it decreases in expectation as $t$ increases.

If $D=n$, the event $\cE_L$ is impossible; so assume $D<n$ and put
$q=D/n$.

We condition on the state after time $t$,
i.e., on the first $t$ exposed vertices.
If $\cE_t$ holds,
the next unexposed vertex lies in $N(y)\cap R_t$ with probability
$p_t=\frac{B_t}{n-t}>\frac{DA_t}{n}=qA_t$
(all vertices in $R_t$  are not yet exposed).
In that case we have $A_{t+1}=A_t+1$ (that vertex gets selected)
and $Z_{t+1} \le 1/(A_t+1)$,
otherwise $A_{t+1}=A_t$
and $Z_{t+1} \le 1/A_t$.
(If $\cE_t$ does not hold, then $Z_{t+1} = 0$
because $A_t$ is nondecreasing and $B_t$ is nonincreasing.)

Thus,
writing $\E_t$ for conditional expectation after time $t$, we have
\[
  \E_t[Z_{t+1}]
  \le\mathbf 1_{\cE_t}\cdot\left(p_t \cdot \frac{1}{A_t+1}+(1-p_t) \cdot \frac{1}{A_t}\right)
\]
and
since the bound
is decreasing in $p_t$, we may replace $p_t$ by the smaller $qA_t$:
\[
  ... \le
  \mathbf 1_{\cE_t}\cdot\left(\frac{qA_t}{A_t+1}+\frac{1-qA_t}{A_t}\right)
  =\mathbf 1_{\cE_t}\cdot\left(\frac1{A_t}-\frac{q}{A_t+1}\right)
  \le(1-q/2)Z_t .
\]
Then, by induction,
we have
\begin{align*}
\E[Z_L] &= \E[\E_{L-1}[Z_L]] \le (1-q/2)\E[Z_{L-1}] \le \\ &\le ... \le (1-q/2)^L \E[Z_0] \le e^{-qL/2} \E[Z_0] \le e^{-qL/2} = e^{-DL/(2n)} = e^{-15\log n} = n^{-15} .
\end{align*}
Now, $Z_L$ is either 0 or at least $1/n$,
and the latter happens iff $\cE_L$ holds, so by Markov,
\[ \Prb(\cE_L) = \Prb(Z_L\ge 1/n)\le \E[Z_L] / (1/n) \le n^{-14} .
\]
A union bound over the $n$ choices of $y$
gives the claimed failure probability $n^{-13}$.

Finally, for $y\in R$ we have $a_y=0$, so
$b_y=\deg_{G[R]}(y)\le D$.
\end{proof}

\subsection{The algorithm}
\label{sec:implicit}

Now we are ready to give the rest of the algorithm.
For clarity we also present it as pseudocode -- see \crefrange{alg:trial}{alg:queryr}.

A single run of our algorithm
(\textsf{Trial} -- \cref{alg:trial})
builds the prefix
(\textsf{Prefix} -- \cref{alg:prefix})
and then performs inverse sampling using
a membership oracle
(i.e., it draws/samples root vertices $x \in V$ and checks whether $x$ is in the MIS, until $\tO(1)$ successes, then uses the number of draws to estimate the MIS size).

This membership oracle
(\textsf{Member} -- \cref{alg:member})
first handles the cases $x \in P$ and $x \not \in R$;
if $x \in R$, it then launches the \emph{residual recursion} oracle
($\Query_R$ -- \cref{alg:queryr}).

The algorithm never materializes $R$.
For a suffix vertex $y\in V\setminus P$,
whenever we need to check whether $y \in R$,
we call the \emph{residual test} $\Residual(y)$
(\cref{alg:residual}),
which
examines the vertices of $S$ in a fresh uniformly
random order and returns \textbf{false} at the first neighbor,
\textbf{true} only if none is found; thus $\Residual(y) =$ \textbf{true} iff $y\in R$.

For a vertex $x$ already known to lie in $R$, define $\Query_R(x)$ (\cref{alg:queryr})
recursively: scan the suffix positions
$L+1,\dots,\rank_\pi(x)-1$
(we call these \emph{raw} positions to differentiate them from residual positions).
For the vertex $y$ in a scanned position, probe
whether $(x,y)\in E$, and if so run $\Residual(y)$, recursing on $y$ only if
the test returns \textbf{true}; return \textbf{false} as soon as a
recursive call returns \textbf{true}, and \textbf{true} if the scan
finishes.

Finally, to obtain our running time bound deterministically,
we run $\tO(1)$ copies of the algorithm
(i.e., of \textsf{Trial})
to amplify the runtime guarantee.
This is routine and we defer it to \cref{sec:accuracy};
our final estimator is \cref{alg:estimate} there.

\begin{figure}[p]
\small
\setlength{\intextsep}{2pt}
\begin{algorithm}[H]
\caption{One \emph{trial}: fresh permutation, fresh prefix, then inverse sampling.}
\label{alg:trial}
\begin{algorithmic}[1]
\Procedure{Trial}{}
  \State $Q\gets\Theta(\eps^{-2}\log(2n))$, \ $\delta\gets\eps/3$
  \State draw a uniform permutation $\pi$ of $V$
  \State $S\gets{}$\Call{Prefix}{$\pi$}
  \State $\mathsf{draws}\gets 0$, \ $\mathsf{successes}\gets 0$
  \While{$\mathsf{successes}<Q$}
    \State draw $x$ uniformly from $V$; \ $\mathsf{draws}\gets\mathsf{draws}+1$
    \If{\Call{Member}{$x$}}
      \State $\mathsf{successes}\gets\mathsf{successes}+1$
    \EndIf
  \EndWhile
  \State \Return $\widehat M=\min\{n,\ nQ/((1-\delta)\,\mathsf{draws})\}$
\EndProcedure
\end{algorithmic}
\end{algorithm}

\begin{algorithm}[H]
\caption{The random prefix of \cref{sec:prefix}.}
\label{alg:prefix}
\begin{algorithmic}[1]
\Procedure{Prefix}{$\pi$}
  \State $L\gets\lceil n^{2/3}\rceil$;
  $P\gets\{\pi(1),\dots,\pi(L)\}$
  \State probe $(u,v)$ for every pair $u,v\in P$
    \Comment{\emph{prefix probes}: $O(L^2)$}
  \State $S\gets\varnothing$
  \For{$i=1,\dots,L$}
    \If{$\pi(i)$ has no neighbor in $S$}
      \State $S\gets S\cup\{\pi(i)\}$
    \EndIf
  \EndFor
  \State \Return $S$
\EndProcedure
\end{algorithmic}
\end{algorithm}

\begin{algorithm}[H]
\caption{The membership oracle.}
\label{alg:member}
\begin{algorithmic}[1]
\Procedure{Member}{$x$}
  \If{$x\in S$} \State \Return \textbf{true} \EndIf
  \If{$x\in P \setminus S$} \State \Return \textbf{false} \EndIf
  \If{\textbf{not} \Call{Residual}{$x$}}
      \Comment{\emph{top-level $\Residual$ test}}
    \State \Return \textbf{false}
  \EndIf
  \State \Return \Call{Query$_R$}{$x$}
    \Comment{$x\in R$: enter the residual recursion at \emph{root} $x$}
\EndProcedure
\end{algorithmic}
\end{algorithm}

\begin{algorithm}[H]
\caption{The survival test of \cref{sec:implicit}: is $y\in R$?
Uses fresh randomness every time (see \cref{lem:cert-cost} for implementation details).}
\label{alg:residual}
\begin{algorithmic}[1]
\Procedure{Residual}{$y$}
  \For{$v\in S$ in a fresh uniformly random order}
    \If{$(y,v)\in E$} \State \Return \textbf{false} \EndIf
  \EndFor
  \State \Return \textbf{true}
    \Comment{$y\in R$}
\EndProcedure
\end{algorithmic}
\end{algorithm}

\begin{algorithm}[H]
\caption{The residual recursion: random-greedy membership inside $G[R]$,
but run over the raw suffix since we do not materialize $R$.}
\label{alg:queryr}
\begin{algorithmic}[1]
\Procedure{Query$_R$}{$x$} \Comment{precondition: $x\in R$}
  \For{$j=L+1,\dots,\rank_\pi(x)-1$}
    \Comment{\emph{raw suffix scan}}
    \State $y\gets\pi(j)$
    \If{$(x,y)\in E$} \Comment{one \emph{raw suffix probe} per position $j$}
      \If{\Call{Residual}{$y$}}
        \Comment{$\Residual$ test at a scanned neighbor}
        \If{\Call{Query$_R$}{$y$}} \State \Return \textbf{false}
          \EndIf
      \EndIf
    \EndIf
  \EndFor
  \State \Return \textbf{true}
\EndProcedure
\end{algorithmic}
\end{algorithm}
\end{figure}

We begin the runtime analysis from the $\Residual$ test.

\begin{lemma}[cost of a $\Residual$ test]
\label{lem:cert-cost}
Conditioned on the prefix, $\Residual(y)$ can be implemented to run in
expected time $O\big(|S|/(a_y+1)\big)$.
\end{lemma}
\begin{proof}
If $a_y\ge 1$, the number of inspected elements is the position of the
first marked element in a uniform arrangement of $|S|$ elements with $a_y$
marked, whose expectation is $(|S|+1)/(a_y+1)$.  If $a_y=0$, the test inspects
all $|S|$ elements.

It remains to see that an inspected element costs only $O(1)$ time, i.e.\
that the random order need not be generated in advance.  Keep one static array
$A[1..|S|]$ holding $S$, built once per $\mathsf{Trial}$ at cost $O(|S|)=O(L)$,
which is negligible.  A test produces its
order by a partial Fisher--Yates shuffle: at step $i$ it draws a uniform
$j\in\{i,\dots,|S|\}$, swaps $A[i]$ with $A[j]$, and inspects $A[i]$.  The
first $k$ elements produced this way are distributed exactly as the first
$k$ elements of a uniformly random order of $S$, so a test that stops after
$k$ inspections never has to touch the rest.  To leave $A$ ready for the
next test, the test remembers the indices $j$ it drew and, on return, undoes
its swaps in reverse order.
\end{proof}

The next two lemmas bound how often the recursion tests a given suffix
vertex, and how many positions one evaluation scans in total.

\begin{lemma}[$\Residual$ tests of a suffix vertex]
\label{lem:calls}
Fix a prefix with $R\ne\varnothing$, and fix a suffix vertex
$y$ with $b_y=|N(y)\cap R|$.  Evaluate $\Query_R$ at a uniform root of $R$,
with the order of $R\cup\{y\}$ uniform.  Then the expected number of $\Residual(y)$ calls made
during that evaluation, counted with multiplicity, is $O(b_y/|R|)$.
\end{lemma}
\begin{proof}
A call $\Residual(y)$ is issued exactly when the scan of some occurrence of
$\Query_R(x)$ reaches the position of $y$ and finds $(x,y)\in E$; such an $x$
lies in $N(y)\cap R$, so what we must count is how often the recursion
follows one of the $b_y$ edges $(x,y)$ with $x\in N(y)\cap R$.

If $y\in R$, the $\Residual$ test is immediately followed by the recursive call
$\Query_R(y)$, so these are exactly the traversals of the $b_y$ 
edges of $G[R]$ entering $y$ made by the standard oracle on $G[R]$ under $\tau$.
\Cref{lem:published}(i) bounds the expected multiplicity of each edge by
$1/|R|$; summing over the $b_y$ edges proves the claim.

Let now $y\notin R$, so that the test returns \textbf{false} and the
recursion never enters $y$.  Then \cref{lem:published} does not apply to $G[R]$,
and we instead realize ``$y$ is always rejected'' inside a larger graph where
it does. This has some similarity to the proof of Lemma 3.8 in~\cite{MRTV26} but is simpler.

Augment $G[R]$ by $y$ and a set of $2(|R|+1)$ \emph{guard vertices}; keep the $b_y$
edges from $y$ to $R$, make every guard adjacent to $y$, and add no other
guard edges.  The augmented graph has $3|R|+3$ vertices.  Draw a uniform
permutation and an independent uniform root, and condition on the \emph{good event}
that the first vertex is a guard and the root lies in $R$. The good event has
probability $\frac{2(|R|+1)}{3|R|+3}\cdot\frac{|R|}{3|R|+3}\ge\frac19$.
Moreover, conditioned on the good event,
the root is uniform in $R$ and the induced order on $R\cup\{y\}$ is uniform.
The first guard is selected into the MIS; being adjacent to $y$ and preceding everything,
it makes every evaluation of $y$ return \textbf{false} at once, exactly as
$\Residual(y)$ does.  Since guards have no neighbors in $R$,
the standard oracle run on the augmented graph
(under the good event) makes the same edge traversals
as the actual algorithm, except that the traversals of edges $(x,y)$ in the augmented graph
correspond to the $\Residual(y)$ calls of the actual algorithm.
For each of
those $b_y$ edges,
\cref{lem:published}(i) (applied to the augmented graph) gives expected
multiplicity $O(1/|R|)$. Finally, conditioning on an event of probability at least
$1/9$ inflates this by at most $9$
(in other words, if the expectation under the good event -- which we want to bound -- were too large, the overall expectation
-- bounded by \cref{lem:published}(i) -- would have also been too large).
Summing over the $b_y$ edges proves the
claim.
\end{proof}

\begin{lemma}[raw adjacency matrix probes]
\label{lem:scan}
Fix a prefix with $R\ne\varnothing$.  Conditioned on it, draw the
suffix order uniformly and an independent uniform root $x\in R$.  Then the
expected total number of raw suffix positions scanned by $\Query_R(x)$,
over all its recursive calls and counted with multiplicity, is $O(n)$.
\end{lemma}
\begin{proof}
Conditioned on the prefix, the order $\tau$ of $R$ is uniform and independent
of the set of suffix positions that $R$ occupies: generate the suffix order
by first choosing the $|R|$ positions of $R$, then independent uniform orders
inside $R$ and inside its complement. 
The evaluation restricted to $R$ is the standard oracle of \cref{sec:oracle}
run on $G[R]$ under $\tau$; let its call occurrences examine
$k_1,\dots,k_A$ residual positions, and put $K=\sum_h(k_h+1)$.
Here $A$ is the number of calls, and we have $K = A + \sum_h k_h$.

\textbf{Bounding the number of \emph{residual} positions scanned.} 
\Cref{lem:published}(ii), applied to $G[R]$ with the
uniform order $\tau$ and the uniform root $x$, bounds $\E[\sum_hk_h]$ by
$O(|R|)$.  As for the number of calls,
every call occurrence but the initial one is entered along an
edge of $G[R]$, and by \cref{lem:published}(i) each of the at most
$\binom{|R|}{2}$ edges is traversed at most $1/|R|$ times in expectation, so
$\E[A]=O(|R|)$.  Adding the two bounds, $\E[K]=O(|R|)$.

\textbf{Bounding the number of \emph{raw} positions scanned.}  Write $N=n-L$ for the length of the suffix and
$Z_j$ for the suffix position of the $j$-th vertex of $R$.  Distributing the
$N-|R|$ discarded positions into the $|R|+1$ gaps around the surviving ones,
each gap has expected size $(N-|R|)/(|R|+1)$ by symmetry, so
$\E[Z_j]=j+j\frac{N-|R|}{|R|+1}=\frac{j(N+1)}{|R|+1}\le\frac{jn}{|R|}$.  A
call examining $k_h$ residual positions stops no later than just before its
$(k_h+1)$-st one --- it either stops at an examined residual neighbor or
reaches its own position, which it never examines, so $k_h+1\le|R|$ --- and
therefore scans at most $Z_{k_h+1}-1$ raw positions.  The positions being
independent of $\tau$ and of $x$, summing over the calls gives
\[
  \E\big[\text{raw positions scanned}\,\big|\,\tau,x\big]
  \le\frac n{|R|}\sum_h(k_h+1)=\frac n{|R|}K.
\]
Averaging over $\tau$ and $x$ and inserting $\E[K]=O(|R|)$ gives $O(n)$.
\end{proof}

The following lemma bounds the number of root-vertex draws
in a single \textsf{Trial}
-- the larger the MIS, the fewer draws we need.
It also uses the observation that $|R|$ cannot be too large compared to the MIS size,
since $G[R]$ has small maximum degree.

\begin{lemma}[number of roots drawn per $\mathsf{Trial}$]
\label{lem:num_roots}
Conditioned on any execution permutation $\pi$,
in one run of $\mathsf{Trial}$,
the expected number of:
\begin{itemize}
    \item
        all roots drawn is $Qn/|\RGMIS(\pi)|$,
    \item
        $R$-roots drawn is $Q|R|/|\RGMIS(\pi)|$,
        which, under the event of \cref{lem:reldeg}, is $O(QD)$.
\end{itemize}
\end{lemma}
\begin{proof}
We show that
for any $x \in V$,
the expected number of times $x$ is drawn as root
is $Q/|\RGMIS(\pi)|$.
By linearity of expectation, it is enough to show that the expected number of draws of $x$
until the first success is $1/|\RGMIS(\pi)|$.
Denote this number by $X$, and the first vertex drawn by $f$.
We have the recursion
\[ \E[X] = \Prb(f=x) + \Prb(f \not \in \RGMIS(\pi)) \cdot \E[X] = \frac{1}{n} + \frac{n-|\RGMIS(\pi)|}{n} \E[X] , \]
which solves to $\E[X] = 1/|\RGMIS(\pi)|$.
The first part follows by summing over all $x \in V$.

For the second part, the first bound follows by summing over all $x \in R$.
For the second bound,
note that under the event of \cref{lem:reldeg},
$\RGMIS_{G[R]}(\tau)$ is maximal in
a graph of maximum degree at most $D$,
so its closed neighborhoods cover $R$.
Therefore
\[
  (D+1) |\RGMIS_{G[R]}(\tau)| \ge |R|
\]
and thus
\[
  \frac{Q|R|}{|\RGMIS(\pi)|} \le \frac{Q|R|}{|\RGMIS_{G[R]}(\tau)|} \le Q(D+1) .
\]
\end{proof}

Now we put all the pieces together.

\begin{proposition}[work of one uncapped $\mathsf{Trial}$]
\label{prop:work}
Conditioned on any prefix satisfying \cref{lem:reldeg}, the expected
runtime of one $\mathsf{Trial}$ (sampling until $Q$ successes) is
$O(L^2+QDn)$.
\end{proposition}
\begin{proof}
It should be clear that the running time of the entire algorithm
is dominated by:
\begin{itemize}
    \item the bulk prefix probes in $\mathsf{Prefix}$,
    \item the raw suffix scans in $\Query_R$,
    \item the $\Residual$ tests,
    which arise in two places: top-level $\Residual$ tests (made directly from $\Member$),
    and $\Residual$ tests at scanned neighbors (made from $\Query_R$).
\end{itemize}
We treat them one by one.

\paragraph{Bulk prefix probes.}
These cost $O(L^2)$.

\paragraph{Raw suffix scans in $\Query_R$.}
\cref{lem:scan} says that
for one random root $x \in R$,
the expected number of raw suffix positions scanned
by all the recursive calls of $\Query_R(x)$
is $O(n)$.
By \cref{lem:num_roots},
conditioned on the permutation $\pi$,
the expected number of random roots in $R$ per $\mathsf{Trial}$
is $O(QD)$;
note that this bound does not depend on $\pi$,
so bounding pointwise gives us a total of $O(QDn)$.

\paragraph{Top-level $\Residual$ tests.}
By \cref{lem:num_roots} there are
at most $Qn/|\RGMIS(\pi)| \le Qn/|S|$ many in expectation,
and each costs $O(|S|)$ by \cref{lem:cert-cost},
for a total of $O(Qn)$.

\paragraph{$\Residual$ tests at scanned neighbors.}
Fix any suffix vertex $y$.
\cref{lem:calls} says that
for one random root $x \in R$,
the expected number of $\Residual$ calls
(made by all the recursive calls of $\Query_R(x)$)
is $O(b_y / |R|)$.
By \cref{lem:num_roots},
conditioned on the permutation $\pi$,
the expected number of random roots in $R$ per $\mathsf{Trial}$
is $Q|R| / |\RGMIS(\pi)| \le Q|R|/|S|$;
note that this bound does not depend on the suffix order,
so bounding pointwise gives us a total of $O(Q b_y / |S|)$
$\Residual$ calls per $\mathsf{Trial}$.
\cref{lem:cert-cost} bounds the cost of a single $\Residual(y)$ test
by $O\big(|S|/(a_y+1)\big)$
(recall that we use fresh randomness there).
In the end, the cost becomes
\[
    O \left(\sum_y \frac{Q b_y}{|S|} \cdot \frac{|S|}{a_y+1}\right) \le 
    O \left(\sum_y QD\right) = O(QDn)
\]
where we used $b_y \le D (a_y+1)$ from \cref{lem:reldeg}.
\end{proof}

Taking
\[
  L=\lceil n^{2/3}\rceil,\qquad
  Q=\Theta(\eps^{-2}\log(2n)),\qquad
  D=O\!\big(n\log n/L\big),
\]
\cref{prop:work} gives the balanced bound
\begin{equation}
  \label{eq:balanced}
  \E[\text{work of a trial}\mid\text{good prefix}]
  =O\!\big(n^{4/3}+\eps^{-2}n^{4/3}\log^2(2n)\big)
  =\tO(\eps^{-2}n^{4/3}).
\end{equation}

\subsection{Accuracy and worst-case capped running time}
\label{sec:accuracy}

The remaining arguments are routine.
Conditioned on its execution permutation $\pi$, each sampled membership bit is
Bernoulli with mean $p=|\RGMIS(\pi)|/n$, and $\mathsf{draws}$ counts the draws to the
$Q$-th success.

\begin{lemma}[inverse-sampling concentration]
\label{lem:conc}
For every $\delta\in(0,\tfrac12)$,
\[
  \Prb\!\big(|Q/\mathsf{draws}-p|>\delta p\big)
  \le 2e^{-\Omega(\delta^2Q)}.
\]
\end{lemma}
\begin{proof}
$\{\mathsf{draws}\le t\}$ is the event of at least $Q$ successes in the first $t$
Bernoulli trials, and $\{\mathsf{draws}>t\}$ that of fewer than $Q$.  Apply the upper-tail
Chernoff bound at $t=\lfloor Q/((1+\delta)p)\rfloor$ and the lower-tail bound
at $t=\lfloor Q/((1-\delta)p)\rfloor$; each probability is
$e^{-\Omega(\delta^2Q)}$.  If $p=1$ then $\mathsf{draws}=Q$.
\end{proof}

Set $\delta=\eps/3$ and $\widehat M=\min\{n,\ nQ/((1-\delta)\mathsf{draws})\}$.  On the
event of \cref{lem:conc},
\[
  |\RGMIS(\pi)|\le\widehat M\le\frac{1+\delta}{1-\delta}|\RGMIS(\pi)|
  \le(1+\eps)|\RGMIS(\pi)|,
\]
the last step holding for $\eps\in(0,1)$; clipping at $n$ preserves the
lower bound as $|\RGMIS(\pi)|\le n$.

\begin{algorithm}[H]
\caption{The estimator of \cref{thm:mis}.}
\label{alg:estimate}
\begin{algorithmic}[1]
\Procedure{Estimate}{$G,\eps$}
  \State $A\gets\Theta(\log(2n))$
  \For{$i=1,\dots,A$}
    \State run \Call{Trial}{}, aborting it (a \emph{capped}
           trial) as soon as its probe counter or its step counter exceeds a
           large constant times the bound of \cref{prop:work}
    \If{the trial was not aborted}
      \State \Return its output $\widehat M$
    \EndIf
  \EndFor
  \State \Return $n$
    \Comment{failure; probability $n^{-\Omega(1)}$}
\EndProcedure
\end{algorithmic}
\end{algorithm}

\begin{proof}[Proof of \cref{thm:mis}]
The case $n=1$ is immediate, so assume $n\ge2$.
Give each trial a runtime counter, and abort
it once it exceeds a sufficiently large constant multiple of the
corresponding conditional-expectation bound from \cref{prop:work} and
\eqref{eq:balanced}.  By Markov, a
trial with a good prefix completes before the cap with probability at
least $2/3$.

Run $A=\Theta(\log(2n))$ independent capped trials, with a sufficiently
large constant, and return the first that completes (outputting a fixed
permutation and $n$ if none does).  By \cref{lem:reldeg}, each trial
completes with constant probability, so the probability that none completes
is at most $n^{-20}$.  Choose the constant hidden in
$Q=\Theta(\eps^{-2}\log(2n))$ large enough that, by \cref{lem:conc} with
$\delta=\eps/3$ and a union bound over the $A$ trials, the probability that
any trial is inaccurate is also at most $n^{-20}$.  This union bound does not
require completion and accuracy to be independent.  Thus the total failure
probability is at most $2n^{-20}\le n^{-10}$.

Every trial starts from a uniform $\pi$; the first to finish need not be
uniform, but by \cref{lem:decomp} every possible winning permutation defines
a genuine maximal independent set, and the concentration analysis holds for
every fixed permutation.

The caps make the work deterministic: one trial uses
$\tO(\eps^{-2}n^{4/3})$ runtime,
and the $A=O(\log(2n))$ trials contribute another logarithm.
This proves \cref{thm:mis}.
\end{proof}

\subsection{Consequence for Steiner Forest}
\label{sec:mis-sf}

\begin{proof}[Proof of \cref{cor:sf-via-mis}]
The reduction of~\cite{MRTV26} estimates the metric Steiner Forest cost by
estimating, at $O(\log k)$ distance scales, the size of a maximal
independent set of an auxiliary graph given by an adjacency-matrix oracle
(one distance query per adjacency query), and its analysis uses only that
each such set is independent and maximal.  Let $m=O(\log k)$ be the number
of calls.  Apply the construction in the proof of \cref{thm:mis} with fixed
accuracy (say $\eps=\tfrac1{100}$) and constants chosen so each call fails
with probability at most $n^{-20}$.  A union bound over the $m$ calls gives
the claimed probability with the same $\tO(n^{4/3})$ running time.
\end{proof}

\section{Steiner Forest via Capped Spanning Trees}
\label{sec:sf}

Throughout this section $(V,d)$ is the input metric with $|V|=n$, and the algorithm can query entries $d(u,v)$.  The given $k$
demand pairs are denoted by $(s_1,t_1),\dots,(s_k,t_k)\in V\times V$, and we use $\OPT$ to denote the cost of the optimal Steiner Forest.
We may assume $k\ge 1$ and that
$X:=\max_j d(s_j,t_j)>0$, as otherwise $\OPT=0$ and the algorithm outputs
$0$.

 As in~\cite{MRTV26} we assume $k=O(n)$,\footnote{This is without loss of
generality in the sense that every instance is characterized by at most
$n-1$ terminal pairs, but reading the demand list and computing such a
succinct representation costs $O(k)$. Note that some dependence on $k$ is
necessary: consider a single point $v$ placed very far from all
others and appearing in exactly one demand pair.  An algorithm that never
inspects that pair is off by an arbitrarily large factor.} and thus reading
the demand list costs $O(k)=O(n)$.  We assume that distinct points
of $V$ are at positive distance.  Throughout,
$\tO(\cdot)$ hides factors polylogarithmic in $n$, and ``with high
probability'' means with probability at least $1-n^{-10}$.
(The constant $10$ could be made arbitrarily large.)

\subsection{The metric-MST black box}
Our Steiner Forest algorithm uses the following estimator as a black box.
For a metric $(Y,\mu)$ we write $\MST(Y,\mu)$ for its
minimum-spanning-tree weight.

\begin{theorem}[Czumaj and Sohler~\cite{CzumajSohler09}]
\label{thm:cs}
For every $\eps\in(0,1)$ there is an algorithm that, given query access to
the distance matrix of an $N$-point metric $(Y,\mu)$, computes a value
$\mathcal M$ with
\[
  (1-\eps)\,\MST(Y,\mu)\;\le\;\mathcal M\;\le\;(1+\eps)\,\MST(Y,\mu)
\]
with probability at least $3/4$, in $\tO(N\cdot\eps^{-O(1)})$ time and
queries.
\end{theorem}

The success probability $3/4$ can be amplified in the standard way: running
$\Theta(\log(1/\eta))$ independent copies and taking the median yields a
$(1\pm\eps)$-estimate with probability $1-\eta$, at a
$\Theta(\log(1/\eta))$ factor in the running time.  We will invoke this
with $\eps$ a fixed constant and $\eta = n^{-\Theta(1)}$.

\subsection{The estimator}
\label{sec:sf-estimator}

Let us start by setting the following values
\[
  \delta=\frac{X}{2k},\qquad
  L=\big\lfloor\log_2(2k)\big\rfloor,\qquad
  \tau_i=2^i\delta\ \ (0\le i\le L),\qquad
  \tau_{L+1}=2\tau_L.
\]
Call a demand \emph{short} if its length is at most $\delta$, and
\emph{retained} otherwise.  There are at most $k$ short demands, so buying
them directly costs at most $k\delta=X/2$; we set them aside and reconnect
them at the end.

For $0\le i\le L$, let
\[
  A_i=\big\{\,\text{endpoints of retained demands $(s_j,t_j)$ of length } d(s_j,t_j)\ge\tau_i\,\big\}\subseteq V .
\]
Since $\tau_0\le\tau_1\le\dots\le\tau_L$, we have
$A_0\supseteq A_1\supseteq\dots\supseteq A_L$.  Define the
\emph{capped metric} at scale $\tau_i$ by
\[
  \rho_i(u,v)=\min\{d(u,v),\tau_i\},
\]
let $W_i=\MST(A_i,\rho_i)$ (with $W_i=0$ if $|A_i|\le 1$), and set
\[
  B=\max_{0\le i\le L} W_i .
\]
The estimator will output a fixed scaling of an estimate of $B$.  The next
lemma is the two-line check promised in the overview.

\begin{lemma}[capping preserves the metric]
\label{lem:cap-metric}
For every $t>0$, the function $\rho(u,v)=\min\{d(u,v),t\}$ is a metric on
$V$.
\end{lemma}
\begin{proof}
Non-negativity, symmetry, and $\rho(u,v)=0\Leftrightarrow u=v$ are
inherited from $d$ (using $t>0$).  For the triangle inequality, fix
$u,v,w$.  If $d(u,w)\ge t$ or $d(w,v)\ge t$, then
$\rho(u,w)+\rho(w,v)\ge t\ge\rho(u,v)$.  Otherwise both are below $t$, and
$\rho(u,v)\le d(u,v)\le d(u,w)+d(w,v)=\rho(u,w)+\rho(w,v)$.
\end{proof}

We will use one elementary fact repeatedly: a subset of the vertices of a
tree spans a tree of at most twice the cost.

\begin{lemma}[shortcutting]
\label{lem:shortcut}
Let $(Y,\mu)$ be a metric, let $T$ be a tree whose vertices lie in $Y$ and
whose edges are weighted by $\mu$, and let $U$ be a subset of the vertices
of $T$.  Then $U$ has a spanning tree, in the metric $\mu$, of weight at
most $2\,\mu(T)$, where $\mu(T)$ is the total weight of $T$.  In
particular, $\MST(U,\mu)\le 2\,\mu(T)$.
\end{lemma}
\begin{proof}
If $|U|\le 1$, the claim is immediate, so assume $|U|\ge 2$.
Double every edge of $T$ to obtain an Eulerian multigraph, and take an
Euler tour; it is a closed walk visiting every vertex of $T$, of total
weight $2\,\mu(T)$.  Delete from the tour all but the first occurrence of
each vertex of $U$; by the triangle inequality for $\mu$, shortcutting can
only decrease the weight, so we obtain a closed walk on $U$ of weight at
most $2\,\mu(T)$.  Removing one edge leaves a spanning path, hence a
spanning tree, of $U$ of weight at most $2\,\mu(T)$.
\end{proof}

\subsection{The upper bound: every capped MST is small}
\label{sec:sf-upper}

\begin{lemma}
\label{lem:sf-upper}
For every $0\le i\le L$ we have $W_i\le 3\,\OPT$, and hence $B\le 3\,\OPT$.
\end{lemma}
\begin{proof}
If $|A_i|\le 1$, then $W_i=0$ and the claim is immediate, so assume
$|A_i|\ge 2$.
Fix an optimal forest $F^\star$ of cost $\OPT$.  Let $F_1,\dots,F_p$ be the
connected components of $F^\star$ that contain at least one point of $A_i$,
with costs $C_1,\dots,C_p$. Since these are distinct components of a
forest, $\sum_{r} C_r\le\OPT$.

Each $F_r$ contains a point $a\in A_i$, which is an endpoint of a retained
demand of length at least $\tau_i$. We know that its partner is connected to $a$ in
$F^\star$, and hence it lies in the same component $F_r$.  Thus $F_r$ connects two
points at distance at least $\tau_i$. So $C_r\ge\tau_i$, and therefore
\begin{equation}
  \label{eq:p-tau}
  p\,\tau_i\;\le\;\sum_{r=1}^p C_r\;\le\;\OPT .
\end{equation}

Now build a spanning tree of $A_i$ in the metric $\rho_i$.  The sets
$A_i\cap V(F_r)$ partition $A_i$.  Applying \cref{lem:shortcut} to the
metric $d$, the tree $F_r$, and the subset $A_i\cap V(F_r)$ gives a
spanning tree of $A_i\cap V(F_r)$ of $d$-weight at most $2C_r$. Since
$\rho_i\le d$, its $\rho_i$-weight is also at most $2C_r$.  Joining these
$p$ trees by $p-1$ edges, each of $\rho_i$-weight at most $\tau_i$, yields a
spanning tree of $A_i$ of $\rho_i$-weight at most
\[
  \sum_{r=1}^p 2C_r+(p-1)\tau_i\;\le\;2\,\OPT+p\,\tau_i\;\le\;3\,\OPT,
\]
using $\sum_r C_r\le\OPT$ and \eqref{eq:p-tau}.  Hence
$W_i=\MST(A_i,\rho_i)\le 3\,\OPT$.
\end{proof}

\subsection{The lower bound: capped MSTs build a feasible forest}
\label{sec:sf-lower}

For each $i$, fix a maximal subset $U_i\subseteq A_i$ whose points are
pairwise at distance at least $2\tau_i$ (in the original metric $d$), and
put $M_i=|U_i|$.  By maximality, every $a\in A_i$ has a \emph{center}
$c_i(a)\in U_i$ with $d(a,c_i(a))<2\tau_i$.  The next lemma is the bridge
between the packing numbers $M_i$ and the capped MST costs $W_i$.

\begin{lemma}
\label{lem:packing-vs-mst}
For every $0\le i\le L$ we have $M_i\,\tau_i\le 4\,W_i$.
\end{lemma}
\begin{proof}
If $A_i=\emptyset$ then $M_i=0$ and the claim is trivial, so assume
$A_i\ne\emptyset$.  Since retained demands have distinct endpoints, $A_i$
then contains both endpoints of a demand of length at least $\tau_i$. Since their
$\rho_i$-distance is exactly $\tau_i$, any spanning tree of $A_i$ has
$\rho_i$-weight at least $\tau_i$, giving $W_i\ge\tau_i$.  This already
settles the case $M_i=1$: $M_i\tau_i=\tau_i\le W_i\le 4W_i$.

Suppose $M_i\ge 2$.  Distinct points of $U_i$ are at $d$-distance at least
$2\tau_i>\tau_i$, and hence at $\rho_i$-distance exactly $\tau_i$. Therefore
every spanning tree of $U_i$ has $\rho_i$-weight exactly $(M_i-1)\tau_i$.
On the other hand, $U_i\subseteq A_i$, and so applying \cref{lem:shortcut} to
$\rho_i$, we have a minimum spanning tree of $A_i$ (of weight $W_i$) and the subset
$U_i$ gives a spanning tree of $U_i$ of $\rho_i$-weight at most $2W_i$.
Thus $(M_i-1)\tau_i\le 2W_i$, and since $M_i\le 2(M_i-1)$ for $M_i\ge 2$,
\[
  M_i\tau_i\le 2(M_i-1)\tau_i\le 4W_i. \qedhere
\]
\end{proof}

\begin{lemma}
\label{lem:feasible-forest}
There is a Steiner forest connecting all retained demands with cost at most
$2X+6\sum_{i=0}^L M_i\tau_i$.  Consequently
$\OPT\le \tfrac{5}{2}X+6\sum_{i=0}^L M_i\tau_i$.
\end{lemma}
\begin{proof}
For each $u\in U_i$, call the fibers $c_i^{-1}(u)$, the \emph{level-$i$
clusters}. Note that each has diameter less than $4\tau_i$ and its points lie within
$2\tau_i$ of the center $u$.  We construct the forest in three groups of
edges, maintaining the invariant that after processing level $i$, all
points inside each level-$i$ cluster are mutually connected.

\emph{Base ($i=0$).}  For every $a\in A_0$ add the edge $(a,c_0(a))$, of
length less than $2\tau_0$.  This connects each level-$0$ cluster.  As
$|A_0|\le 2k$ and $\tau_0=\delta=X/(2k)$, the total cost is less than
$2\tau_0|A_0|\le 2\cdot\frac{X}{2k}\cdot 2k=2X$.

\emph{Transitions ($i\to i+1$, for $0\le i<L$).}  Assume the level-$i$
clusters are connected.  Form an auxiliary graph whose nodes are the
level-$i$ clusters; for each $v\in A_{i+1}$, add an edge between the
level-$i$ cluster containing $v$ and the one containing $c_{i+1}(v)$ (both
$v$ and $c_{i+1}(v)$ lie in $A_{i+1}\subseteq A_i$).  Take a spanning
forest of this auxiliary graph, and realize each of its edges by the
witness pair $(v,c_{i+1}(v))$, of length $d(v,c_{i+1}(v))<2\tau_{i+1}=4\tau_i$.
Since two points share a level-$(i+1)$ cluster only if their centers
$c_{i+1}(\cdot)$ coincide, connecting the level-$i$ clusters this way makes
every level-$(i+1)$ cluster connected.  The auxiliary graph has $M_i$
nodes, so its spanning forest has fewer than $M_i$ edges, and the cost is
less than $4M_i\tau_i$.

\emph{Demand layers ($0\le i\le L$).}  Assign each retained demand of
length $\ell$ to the unique level $i$ with $\tau_i\le\ell<\tau_{i+1}$ (such
an $i$ exists because retained demands have $\delta=\tau_0\le\ell\le X<\tau_{L+1}$;
the last inequality is shown in \cref{lem:X-vs-B} below).  Both endpoints
then lie in $A_i$, i.e.\ in level-$i$ clusters.  Form an auxiliary graph with nodes being the level-$i$
clusters, and with one edge per level-$i$ demand, joining the clusters of its
two endpoints. Now take a spanning forest, and realize each chosen edge by the
demand itself, of length $\ell<\tau_{i+1}=2\tau_i$.  As the level-$i$
clusters are already connected, this connects both endpoints of every
level-$i$ demand.  There are fewer than $M_i$ edges, and so the cost is less
than $2M_i\tau_i$.

Summing over all levels, the total cost is at most
$2X+\sum_{i}(4M_i\tau_i+2M_i\tau_i)=2X+6\sum_{i=0}^L M_i\tau_i$. Deleting
any cycles only decreases it.  Finally, reconnecting the short demands
directly adds at most $X/2$, giving a feasible forest for all demands of
cost at most $\tfrac{5}{2}X+6\sum_i M_i\tau_i$, which bounds $\OPT$.
\end{proof}

\begin{lemma}
\label{lem:X-vs-B}
$\tau_L\le X<\tau_{L+1}=2\tau_L$, and consequently $X\le 2B$.
\end{lemma}
\begin{proof}
Since $L=\lfloor\log_2(2k)\rfloor$, we have $k<2^L\le 2k$, so
$\tau_L=2^L\delta\le 2k\cdot\frac{X}{2k}=X$ and
$\tau_L=2^L\delta>k\cdot\frac{X}{2k}=X/2$, i.e.\ $X<2\tau_L=\tau_{L+1}$.  The
maximum demand has length $X>\delta$, so it is retained, and since
$X\ge\tau_L$ both of its endpoints lie in $A_L$. Therefore, their $\rho_L$-distance is
$\min\{X,\tau_L\}=\tau_L$, so $W_L\ge\tau_L>X/2$.  Hence $X<2W_L\le 2B$.
\end{proof}

Combining the pieces gives the deterministic sandwich at the heart of the
algorithm.

\begin{proposition}
\label{prop:sandwich}
$\dfrac{\OPT}{24L+29}\;\le\;B\;\le\;3\,\OPT.$
\end{proposition}
\begin{proof}
The upper bound is \cref{lem:sf-upper}.  For the lower bound, combine
\cref{lem:feasible-forest} with $M_i\tau_i\le 4W_i\le 4B$
(\cref{lem:packing-vs-mst}) for each of the $L+1$ levels, and then with $X\le 2B$
(\cref{lem:X-vs-B}):
\[
  \OPT\le\tfrac{5}{2}X+6\sum_{i=0}^L M_i\tau_i
       \le 5B+24(L+1)B=(24L+29)B. \qedhere
\]
\end{proof}

\subsection{The algorithm}
\label{sec:sf-algo}

We can now prove \cref{thm:sf}.

\begin{proof}[Proof of \cref{thm:sf}]
Set $\eps=\tfrac12$.  Compute $X=\max_j d(s_j,t_j)$ with $O(k)$ queries, and
if $X=0$ output $0$. Otherwise, compute $\delta$, $L$, and the scales.  For each
$0\le i\le L$, compute
$A_i\subseteq V$ from the demand list (identifying repeated points). Note that this
costs $O(k)$ per level and $\sum_i|A_i|\le 2k(L+1)$ in total.

For each $i$, obtain an estimate $\widehat W_i$ of $W_i$ as follows.  If
$|A_i|\le 2$, compute $W_i$ exactly with $O(1)$ distance queries.
Otherwise run the estimator of \cref{thm:cs} on $(A_i,\rho_i)$ with the
constant accuracy $\eps=\tfrac12$, amplified by taking the median of
$\Theta(\log n)$ independent runs so that
\[
  (1-\eps)W_i\le\widehat W_i\le(1+\eps)W_i
\]
fails with probability at most $n^{-10}/(L+1)$. Note that a distance query in
$\rho_i$ is one query in $d$ followed by a cap.  Put
$\widehat B=\max_i\widehat W_i$ and output
\[
  Z=\frac{24L+29}{1-\eps}\,\widehat B .
\]

\emph{Correctness.}  By a union bound over the $L+1$ levels, with
probability at least $1-n^{-10}$ every estimate is accurate, and hence
$(1-\eps)B\le\widehat B\le(1+\eps)B$.  On this event, using
\cref{prop:sandwich},
\[
  Z=\frac{24L+29}{1-\eps}\widehat B\ge (24L+29)B\ge\OPT,
\]
and
\[
  Z\le\frac{24L+29}{1-\eps}(1+\eps)B
   =3(24L+29)B\le 9(24L+29)\,\OPT=O(\log k)\cdot\OPT,
\]
where we used $(1+\eps)/(1-\eps)=3$ and $B\le 3\OPT$.  Thus $Z$ is a
one-sided $O(\log k)$-approximation of $\OPT$.

\emph{Complexity.}  By \cref{thm:cs} with constant $\eps$, each level costs
$\tO(|A_i|)$ time and queries, and the median amplification multiplies this
by $\Theta(\log n)$.  Summing, the total is
$\tO\!\big(\sum_i|A_i|\big)=\tO(k(L+1))=\tO(k)$, plus $O(k)$ to read the
demands.
\end{proof}

\section{Open questions}
\label{sec:discussion}

Our work leaves the following two natural questions open:
\begin{itemize}
    \item All the known sublinear algorithms for Steiner Forest give an $O(\log k)$ approximation. Is it possible to improve the approximation ratio to e.g.~$O(1)$ in sublinear time? The $\Theta(\log k)$ loss is inherent to the scale-based estimators used in~\cite{MRTV26} and our work,
    so improving the ratio entails obtaining a genuinely different handle on $\OPT$.
    \item Is there an MIS size estimation algorithm that runs in $\tO(n)$ time?
\end{itemize}

\section*{AI Disclosure}
The algorithms and their analyses were mainly discovered using a custom agentic harness developed by the authors,
which uses GPT 5.6 Sol, Opus 5, and Gemini 3.1 Pro.
All proofs were verified by each author and then largely rewritten by the authors.
The authors assume all responsibility for the paper's content and its correctness.

\bibliographystyle{alpha}
\bibliography{refs}

\end{document}